\documentclass[		copyright,
	creativecommons,
	noderivs,
	noncommercial,
]{eptcs}

\providecommand{\firstpage}{1}

\usepackage{snapshot}
\usepackage{amsmath}
\usepackage{amsthm}

\RequirePackage{xstring}

\let\iftwocolumn\iffalse
\let\ifonecolumn\iftrue

\newtheorem{definition}{Definition}

\newtheorem{theorem}{Theorem}

\newcounter{constrainttype}
\newenvironment{constraints}{	\begin{compactenum}
	\setcounter{enumi}{\value{constrainttype}}
	\let\olditem\item
	\renewcommand{\item}{\olditem \refstepcounter{constrainttype}}
}{	\end{compactenum}
}

\newcommand{\enquote}[1]{``#1''}

\newcommand{\keywords}[1]{\textbf{Keywords:} #1}

\newcommand{\wantsdraftwatermark}{no}
\IfStrEq*{\wantsdraftwatermark}{yes}{	\RequirePackage[draft,grayness=0.85,allpages,xcoord=50,ycoord=-205]{sty/draftmark}
}{}

\newcommand{\numberofauthors}[1]{}
\newcommand{\titlenote}[1]{\thanks{#1}}
\newcommand{\affaddr}[1]{#1}
\newcommand{\category}[3]{}
\newcommand{\terms}[1]{}

\newcommand{\fundinghint}{	This work has been partly supported by the German Research Foundation
	(DFG) as part of the Transregional Collaborative Research Center
	``Automatic Verification and Analysis of Complex Systems'' (SFB/TR 14 AVACS).}
\usepackage{latexsym,amssymb}
\usepackage{amscd} \usepackage{xspace}
\usepackage{caption}
\usepackage{subcaption}
\usepackage{graphicx}
\usepackage{hyperref}
\usepackage[mathscr]{euscript}
\usepackage[disable]{todonotes}
\usepackage{paralist}
\usepackage{tabularx}
\RequirePackage{amsmath}
\RequirePackage{amssymb}
\RequirePackage{stmaryrd} \RequirePackage{xspace}
\RequirePackage{nicefrac}

\let\originalleft\left
\let\originalright\right
\renewcommand{\left}{\mathopen{}\mathclose\bgroup\originalleft}
\renewcommand{\right}{\aftergroup\egroup\originalright}

\newcommand{\ie}{i.\,e.}

\newcommand{\wrt}{wrt.\ }

\newcommand{\SetN}{\ensuremath{\mathbb{N}}}
\newcommand{\SetR}{\ensuremath{\mathbb{R}}}

\newcommand{\PowerSet}[1]{\ensuremath{\mathcal{P}\left( #1 \right)}}
\newcommand{\abs}[1]{\ensuremath{\left| #1 \right|}}
\newcommand{\card}[1]{\ensuremath{|#1|}}
\newcommand{\holds}[2][]{#1 : #2}
\newcommand{\set}[2]{\left\{ #1 \;\middle|\; #2 \right\}}
\newcommand{\innerproduct}[2]{\left\langle #1 \;\middle|\; #2\right\rangle}
\newcommand{\timefrac}[2]{\frac{d#1}{d#2}}
\newcommand{\project}[2]{#1_{\downarrow#2}}
\newcommand{\V}[1]{\mathbf{#1}}

\newcommand{\IdentFunc}{\mathrm{id}}
\newcommand{\ZeroFunc}{\mathtt{zero}}
\newcommand{\IdentOf}[1]{\IdentFunc\left(#1\right)}

\newcommand{\tuple}[1]{\left(#1\right)}

\newcommand{\cupdot}{\mathbin{\mathaccent\cdot\cup}}

\newcommand{\impl}{\Rightarrow}

\newcommand{\NP}{\ensuremath{\mathrm{NP}}}

\RequirePackage{amsmath}
\RequirePackage{amssymb}
\RequirePackage{xspace}
\RequirePackage{nicefrac}

\newcommand{\Stabhyli}{\textsc{Stabhyli}\xspace}
\newcommand{\RSolver}{\textsc{RSolver}\xspace}

\newcommand{\Matlab}{\textsc{Matlab}\xspace}

\newcommand{\SOSTools}{\textsc{SOSTools}\xspace}
\newcommand{\YALMIP}{\textsc{YALMIP}\xspace}

\newcommand{\HyAt}{\ensuremath{\mathcal{H}}}
\newcommand{\SetOfModes}{\ensuremath{\mathcal{M}}}
\newcommand{\Mode}{\ensuremath{m}}
\newcommand{\ContStateSpace}{\ensuremath{\mathcal{S}}}

\newcommand{\SetOfVars}{\ensuremath{\mathcal{V}}}

\newcommand{\SetOfTrans}{\ensuremath{\mathcal{T}}}
\newcommand{\MapOfFlow}{\ensuremath{\mathit{Flow}}}
\newcommand{\FlowOf}[1]{\ensuremath{\MapOfFlow\left( #1 \right)}}

\newcommand{\MapOfInvar}{\ensuremath{\mathit{Inv}}}
\newcommand{\InvOf}[1]{\ensuremath{\MapOfInvar\left( #1 \right)}}

\newcommand{\GuardSet}{\ensuremath{\mathit{G}}}
\newcommand{\UpdateFunc}{\ensuremath{\mathit{U}}}
\newcommand{\UpdateOf}[1]{\ensuremath{\UpdateFunc\left(#1\right)}}
\newcommand{\Trans}[1]{\ensuremath{\left(#1\right)}}
\newcommand{\trans}{t}

\newcommand{\Ally}{\ensuremath{\mathit{ST}}}

\newcommand{\Relax}{\ensuremath{\mathit{Rlx}}}
\newcommand{\RelaxApp}[1]{\ensuremath{\Relax\left(#1\right)}}

\newcommand{\traj}{\ensuremath{\tau}}
\newcommand{\trajApp}[1]{\ensuremath{\traj\left(#1\right)}}
\newcommand{\trajC}{\ensuremath{\V{x}}}
\newcommand{\trajCApp}[1]{\ensuremath{\trajC\left(#1\right)}}
\newcommand{\trajM}{\ensuremath{m}}
\newcommand{\trajMApp}[1]{\ensuremath{\trajM\left(#1\right)}}

\newcommand{\LF}[1][]{\ensuremath{V_{#1}}}
\newcommand{\dLF}[1][]{\ensuremath{\dot{\LF}_{#1}}}
\newcommand{\LFof}[2][]{\ensuremath{\LF[#1]\left(#2\right)}}
\newcommand{\dLFof}[2][]{\ensuremath{\dLF[#1]\left(#2\right)}}

\newcommand{\norm}[1]{\ensuremath{\left|\left|#1\right|\right|}}

\newcommand{\classKinf}{\ensuremath{\mathit{K}^{\infty}}}
\newcommand{\NormOf}[2]{\ensuremath{#1\left(\norm{#2}\right)}}
\newcommand{\NormAlpha}{\ensuremath{\alpha}}
\newcommand{\NormBeta}{\ensuremath{\beta}}
\newcommand{\NormGamma}{\ensuremath{\gamma}}
\newcommand{\NormAlphaOf}[1]{\ensuremath{\NormOf{\NormAlpha}{#1}}}
\newcommand{\NormBetaOf}[1]{\ensuremath{\NormOf{\NormBeta}{#1}}}
\newcommand{\NormGammaOf}[1]{\ensuremath{\NormOf{\NormGamma}{#1}}}

\newcommand{\Graph}{\ensuremath{\mathscr{G}}}
\newcommand{\SetOfVertices}{\ensuremath{\mathscr{V}}}
\newcommand{\SetOfEdges}{\ensuremath{\mathscr{E}}}

\usepackage[
	ruled,
	vlined,
	procnumbered,		linesnumbered,		nokwfunc,				]{algorithm2e}

\let\iflong\iffalse

\newfont{\mycrnotice}{ptmr8t at 7pt}
\newfont{\myconfname}{ptmri8t at 7pt}
\DeclareMathAlphabet{\mathcal}{OMS}{cmsy}{m}{n}

\def\sectionautorefname~{Section\,}
\def\subsectionautorefname~{Section\,}
\def\figureautorefname~{Figure\,}
\def\tableautorefname~{Table\,}
\def\theoremautorefname~{Theorem\,}
\def\definitionautorefname~{Definition\,}
\def\corollaryautorefname~{Corollary\,}
\def\ExampleCounterautorefname~{Example\,}
\def\algorithmautorefname~{Algorithm\,}
\def\functionautorefname~{Function\,}
\def\constrainttypeautorefname~{Type\,}
\def\propositionautorefname~{Proposition\,}

\title{Breaking Dense Structures -- Proving Stability of Densely Structured Hybrid Systems\titlenote{\fundinghint}}

\date{\today}

\numberofauthors{1}
\author{	Eike M\"ohlmann and Oliver Theel\institute{	\affaddr{Carl von Ossietzky University of Oldenburg}\\
	\affaddr{Department of Computer Science}\\
	\affaddr{D-26111 Oldenburg, Germany}\\
	\email{\fontsize{9}{9}\selectfont{			\texttt{\{eike.moehlmann,~theel\}@informatik.uni-oldenburg.de}
	}}
	\vspace*{-1em}
}
}
\def\titlerunning{Breaking Dense Structures}
\def\authorrunning{E.~M\"ohlmann \& O.~Theel}

\begin{document}
\maketitle

\begin{abstract}
\label{sec:abstract}

Abstraction and refinement is widely used in software development. Such
techniques are valuable since they allow to handle even more complex systems.
One key point is the ability to decompose a large system into subsystems,
analyze those subsystems and deduce properties of the larger system. As
cyber-physical systems tend to become more and more complex, such techniques
become more appealing.
\newline
In 2009, Oehlerking and Theel presented a (de-)composition
technique for hybrid systems. This technique is graph-based and
constructs a Lyapunov function for hybrid systems having a complex discrete state
space. The technique consists of
\begin{inparaenum}[(1)]
	\item decomposing the underlying graph of the hybrid system into subgraphs,
	\item computing multiple local Lyapunov functions for the subgraphs, and
		finally
	\item composing the local Lyapunov functions into a piecewise
		Lyapunov function.
\end{inparaenum}
A Lyapunov function can serve multiple purposes, e.g., it certifies stability
or termination of a system or allows to construct invariant sets, which in
turn may be used to certify safety and security.
\newline
In this paper, we propose an improvement to the decomposing technique, which
relaxes the graph structure before applying the decomposition technique. 
Our relaxation significantly reduces the connectivity of
the graph by exploiting super-dense switching. The relaxation makes
the decomposition technique more efficient on one hand and on the other allows
to decompose a wider range of graph structures.

 \end{abstract}
\keywords{	Hybrid Systems,
	Automatic Verification,
	Stability,
	Lyapunov Theory,
	Graphs,
	Relaxation
}
\section{Introduction} \label{sec:introduction}

In this paper, we present a relaxation technique for hybrid systems exhibiting
dense graph structures. It improves the (de-)compositional technique proposed
by
Oehlerking and Theel in \cite{OehlerkingT09}.
The relaxation results in hybrid systems that are
well suited for (de-)composition.
This increases the likeliness of successfully identifying Lyapunov functions.

Throughout the paper, in order to ease readability we will simply write \enquote{decomposition} or \enquote{decompositional technique}
instead of \enquote{(de-)composition} or \enquote{(de-)compositional technique}.

Stability, in general and for hybrid systems in particular, is a very
desirable property, since stable systems are inherently fault-tolerant: after
the occurrence of faults leading to, for example, a changed environment, the
system will automatically \enquote{drive back} to the set of desired (i.e., stable)
states. Stable systems are therefore particularly suited for contexts where
autonomy is important such as for dependable assistance systems or in contexts
where security has to be assured in an adverse environment.

Modeling such real world systems often involves the interaction of embedded
systems (e.g., a controller) and its surrounding environment (e.g., a plant).
Examples of such systems are automatic cruise controllers, engine control
units, or unmanned powerhouses. In all these examples, an optimal operating
range should be maintained. Although it is sometimes possible to discretize
physical relations (using sampling) or to fluidize discrete steps (having a
real-valued count of objects) it is more natural and less error-prone to use
hybrid systems for modeling and verification. This is due to the fact that
hybrid systems allow both: the representation of discrete and continuous
behavior.

For hybrid systems with a complex discrete behavior, the technique proposed in
\cite{OehlerkingT09} decomposes the monolithic problem of proving stability into
multiple subproblems.
But if a hybrid system exhibiting a complex
control structure -- in the sense of a dense graph structure -- is decomposed,
then the blow-up can be enormous. The result is a high number of
subproblems that must be solved -- this is not bad per se. But since the
decompositional technique requires to underapproximate the
feasible sets of each subproblem -- when applied to
often -- results in the feasible set becoming empty. 
The relaxation technique presented in this paper reduces the number of steps
required by the decomposition and, therefore, the number of
underapproximations. This has two benefits: the runtime is reduced
as well as the effect of underapproximations is minimized.

This paper is organized as follows. \autoref{sec:related_work} gives a
brief overview on related work.
In \autoref{sec:preliminaires}, we
define the hybrid system model, the stability property, an adaptation
of the Lyapunov Theorem, and briefly sketch the idea of the decompositional
proof technique.
\autoref{sec:relaxing} describes our improvement to that proof scheme.
In \autoref{sec:experiments}, we apply the relaxation to prove stability of
three examples.
The first example is the automatic cruise controller which
is the motivating example for the decompositional technique. 
The second example is abstract and shows what happens if
decomposition is applied to complete
graph structures.
The last example is a spidercam that exhibits a dense
graph structure for which proving stability using decomposition
is not possible.
Finally, in
\autoref{sec:summary}, we give a short summary.

\section{Related Work} \label{sec:related_work}

In contrast to safety properties, stability has not yet received that much attention
\wrt automatic proving and therefore, only a few tools are available. Indeed
only the following automatic tools -- each specialized for specific system classes --
are known to the authors.
Podelski and Wagner presented a tool in \cite{PodelskiW07} which computes a sequence of snapshots
and then tries to relate the snapshots in decreasing sequence. If
successful, then this
certifies region stability, i.e., stability with respect to a region
instead of a single equilibrium point.
Oehlerking et al.~\cite{OehlerkingBT07} implemented
a powerful state space
partitioning scheme to find Lyapunov functions for linear hybrid
systems.
The \RSolver by Ratschan and She~\cite{RatschanS10} computes Lyapunov-like
functions for continuous system.
Duggirala and Mitra~\cite{DuggiralaM12} combined Lyapunov functions with
searching for a well-foundedness relation for symmetric linear hybrid
systems.
Prabhakar and Garc\'{i}a \cite{PrabhakarS13} presented a technique for proving
stability of hybrid systems with constant derivatives.
Finally, some \Matlab toolboxes (\YALMIP~\cite{YALMIP},
\SOSTools~\cite{sostools}) that require a by-hand generation of constraint
systems for the search of Lyapunov functions are available. These toolboxes do
not automatically prove stability but assist in handling solvers. 

Related theoretical works are the decompositional technique by Oehlerking and
Theel~\cite{OehlerkingT09}, which we aim to improve, and the work on pre-orders for reasoning about
stability in a series of papers by Prabhakar et al. \cite{PrabhakarDV12,Prabhakar12,PrabhakarLM13}
whose aim is a precise characterization of soundness of abstractions for
stability properties. In contrast, our vision is an automatic computational engine for
obtaining Lyapunov functions. The technique and tool presented in \cite{PrabhakarS13}
is also based on abstractions. Unfortunately, their
technique is restricted to hybrid systems whose differential equations have
constant right hand sides while our technique is more general.
However, the techniques are not even mutually exclusive and have the potential
to be combined.

 \section{Preliminaries} \label{sec:preliminaires}

In this section, we give the definitions of the hybrid system model, global
asymptotic stability, and discontinuous Lyapunov functions. Furthermore, we
sketch the decomposition technique of \cite{OehlerkingT09}.

\begin{definition} \label{def:Hybrid_Automaton}
A \textbf{Hybrid Automaton}
\iftwocolumn
is a quintuple
\[ \HyAt = (\SetOfVars, \SetOfModes, \SetOfTrans, \MapOfFlow, \MapOfInvar)
\text{ where} \]
\else
\(\HyAt\)
is a tuple
\((\SetOfVars, \SetOfModes, \SetOfTrans, \MapOfFlow, \MapOfInvar)\)
where
\fi
\begin{compactitem}
	\item \(\SetOfVars\) is a finite set of \emph{variables} and
		\(\ContStateSpace = \SetR^{\abs{\SetOfVars}}\) is the corresponding
		\emph{continuous state space},
	\item \(\SetOfModes\) is a finite set of \emph{modes},
	\item \(\SetOfTrans\) is a finite set of \emph{transitions}
		\(\Trans{\Mode_{1},\GuardSet,\UpdateFunc,\Mode_{2}}\) where
	\begin{compactitem}
		\item \(\Mode_{1}, \Mode_{2} \in \SetOfModes \) are the \emph{source and
			target mode} of the transition, respectively,
		\item \(\GuardSet \subseteq \ContStateSpace\) is a \emph{guard} which
			restricts the valuations of the variables for which this
			transition can be taken,
		\item \(\UpdateFunc : \ContStateSpace \to \ContStateSpace\) is the
			\emph{update function} which might update some valuations of the
			variables,
	\end{compactitem}
	\item \(\MapOfFlow : \SetOfModes \to [\ContStateSpace \to
		\PowerSet{\ContStateSpace}]\) is the \emph{flow function} which
		assigns a \emph{flow} to every mode. A flow \(f \subseteq
		\ContStateSpace \to \PowerSet{\ContStateSpace}\) in turn assigns a
		closed subset of \(\ContStateSpace\) to each \(\V{x} \in
		\ContStateSpace\), which can be seen as the right hand side of a
		differential inclusion \(\dot{\V{x}} \in f(\V{x})\),
	\item \(\MapOfInvar : \SetOfModes \to \PowerSet{\ContStateSpace}\) is the
		\emph{invariant function} which assigns a closed subset of the
		continuous state space to each mode \(\Mode \in \SetOfModes\), and
		therefore restricts valuations of the variables for which this mode
		can be active.
\end{compactitem}
A \emph{trajectory} of \(\HyAt\) is an infinite solution in form of a function
\(\trajApp{t}=(\trajCApp{t}, \trajMApp{t})\) over time 
\(t\) where \(\trajCApp{\cdot}\) describes the evolution 
of the continuous variables and \(\trajMApp{\cdot}\) the corresponding
evolution of the modes.\footnote{	Note, that definition of trajectories given here is for real time, \ie,
	\(t\in\SetR_{\geq 0}\) while solutions of the relaxed hybrid automaton in
	\autoref{sec:relaxing} require a corresponding definition of
	trajectories for dense time, \ie, \(t\in\SetN\times\SetR_{\geq 0}\).
	However, as there is only little difference in our setting and we do not
	directly reason about the solutions of the relaxation, we omit
	corresponding definitions.}
\end{definition} 
Roughly speaking, stability is a property basically expressing that all
trajectories of the system eventually reach an equilibrium point of the
sub-state space and stay in that point forever given the absence of
errors. For technical reasons the equilibrium point is usually assumed to be
the origin of the continuous state space, \ie\ \(\V{0}\). This is not a
restriction, since a system can always be shifted such that the equilibrium is \(\V{0}\)
via a coordinate transformation.
In the sequel, we focus on \emph{asymptotic stability} which does not require the
equilibrium point to be reached in finite time but only requires every
trajectory to \enquote{continuously approach} it (in contrast to \emph{exponential
stability} where additionally the existence of an exponential rate of
convergence is required).

In the following, we refer to \(\project{\V{x}}{\SetOfVars'} \in
\SetR^{\card{\SetOfVars'}}\) as the sub-vector of a vector \(\V{x} \in
\SetR^{\SetOfVars}\) containing only values of variables in \(\SetOfVars'
\subseteq \SetOfVars\).

\begin{definition}[Global Asymptotic Stability with Respect to a Subset of Variables~\cite{Oehlerking2011thesis}] \label{def:GAS}
Let \(\HyAt = (\SetOfVars, \SetOfModes, \SetOfTrans,\linebreak[1] \MapOfFlow,
\MapOfInvar)\) be a hybrid automaton, and let \(\SetOfVars' \subseteq
\SetOfVars\) be the set of variables that are required to converge to the
equilibrium point \(\V 0\). A continuous-time dynamic system \(\HyAt\) is
called \emph{Lyapunov stable (LS)
with respect to \(\SetOfVars'\)}
if for all
functions \(\project{\V{x}}{\SetOfVars'}(\cdot)\),
\[
	\holds[\forall \epsilon\!>\!0 : \: \exists \delta\!>\!0]	{		\holds[\forall t \geq 0]		{			\norm{\V{x}(0)}\!<\!\delta \impl \norm{\project{\V{x}}{\SetOfVars'}(t)}\!<\!\epsilon
		}	}	\text{.}
\]
\(\HyAt\) is called \emph{globally attractive (GA) with respect to \(\SetOfVars'\)}
if for all functions \(\project{\V{x}}{\SetOfVars'}(\cdot)\),
\[
	\lim_{t\to\infty} \project{\V{x}}{\SetOfVars'}(t) = \V{0}
	\text{, \ie,}
	\holds[\forall \epsilon\!>\!0]	{		\holds[\exists t_0\!\geq\!0]		{			\holds[\forall t\!>\!t_{0}]
			{				\norm{\project{\V{x}}{\SetOfVars'}(t)}\!<\!\epsilon
			}		}	}	\text{,}
\]
where \(\V{0}\) is the origin of \(\SetR^{\card{\SetOfVars'}}\). If a system is
both globally stable with respect to \(\SetOfVars'\) and globally attractive
with respect to \(\SetOfVars'\), then it is called \emph{globally asymptotically
stable (GAS) with respect to \(\SetOfVars'\).}
\end{definition} Intuitively, LS is a boundedness condition, \ie, each trajectory starting
\(\delta\)-close
to the origin will remain \(\epsilon\)-close to the origin. GA
ensures progress, \ie, for each \(\epsilon\)-distance to the origin, there exists a
point in time \(t_0\) such that afterwards a trajectory always remains within this
distance.
It follows, that each trajectory is eventually always approaching
the origin.
This property can be proven using Lyapunov Theory~\cite{Lya07}.
Lyapunov Theory was originally restricted to continuous systems
but has been lifted to hybrid systems.
\begin{theorem}[Discontinuous Lyapunov Functions for a subset of variables~\cite{Oehlerking2011thesis}] \label{thm:lyapunov}
Let \(\HyAt = (\SetOfVars,\SetOfModes,\SetOfTrans,\MapOfFlow,\linebreak[1]\MapOfInvar)\) be
a hybrid automaton and let \(\SetOfVars' \subseteq \SetOfVars\) be the set of
variables that are required to converge. If for each \(\Mode \in
\SetOfModes\),
there exists a set of variables \(\SetOfVars_{\Mode}\) with \(\SetOfVars'
\subseteq \SetOfVars_{\Mode} \subseteq \SetOfVars\) and a continuously
differentiable function \(\LF[\Mode] : \ContStateSpace \to \SetR\) such that
\begin{constraints}
	\item \label{thm:lyapunov:mode}
		for each \(\Mode \in \SetOfModes\), there exist two class
		\(\classKinf\) functions \(\NormAlpha\) and \(\NormBeta\) such that
		\[
		\holds[\forall \V{x} \in \InvOf{\Mode}]{
			\NormAlphaOf{\project{\V{x}}{\SetOfVars_{\Mode}}}
			\leq
			\LFof[\Mode]{\V{x}}
			\leq
			\NormBetaOf{\project{\V{x}}{\SetOfVars_{\Mode}}}}
		\text{,}
		\]
	\item \label{thm:lyapunov:flow}
		for each \(\Mode \in \SetOfModes\), there exists a class
		\(\classKinf\) function \(\NormGamma\) such that
		\[
		\holds[\forall \V{x} \in \InvOf{\Mode}]{\dLFof[\Mode]{\V{x}} \leq
		-\NormGammaOf{\project{\V{x}}{\SetOfVars_{\Mode}}}}
		\]
		for each
		\(\dLFof[\Mode]{\V{x}} \in \set{\innerproduct{\timefrac{\LFof[\Mode]{\V{x}}}{\V{x}}}{f(\V{x})}}{f(\V{x}) \in \FlowOf{\Mode}}\),
	\item \label{thm:lyapunov:trans}
		for each \(\Trans{\Mode_{1},\GuardSet,\UpdateFunc,\Mode_{2}} \in
		\SetOfTrans\),
		\[
		\holds[\forall \V{x} \in \GuardSet]{\LFof[\Mode_{2}]{\UpdateOf{\V{x}}} \leq \LFof[\Mode_{1}]{\V{x}}}
		\text{,}
		\]
\end{constraints}
then \(\HyAt\) is globally asymptotically stable with respect to
\(\SetOfVars'\) and \(\LF[\Mode]\) is called a \emph{local Lyapunov function
(LLF)} of \(\Mode\).
\end{theorem} In \autoref{thm:lyapunov}, \(\innerproduct{\timefrac{\LFof{\V{x}}}{\V{x}}}{f(\V{x})}\) denotes
the inner product between the gradient of a Lyapunov function \(\LF\) and a
flow function \(f(\V{x})\).
Throughout the paper we denote by \emph{mode
constraints} the constraints of \autoref{thm:lyapunov:mode}
and \autoref{thm:lyapunov:flow} and by \emph{transition constraints} the
constraints of \autoref{thm:lyapunov:trans}.

\subsection*{Decompositional Construction of Lyapunov Functions}
\label{sec:decomposition}

In this section we briefly introduce the decompositional construction of Lyapunov
functions for self-containment and refer to \cite{OehlerkingT09} for the
details.

\begin{figure} 	\centering
	\hfill
	\begin{subfigure}{0.32\linewidth}
		\centering
		\includegraphics[width=0.8\linewidth]{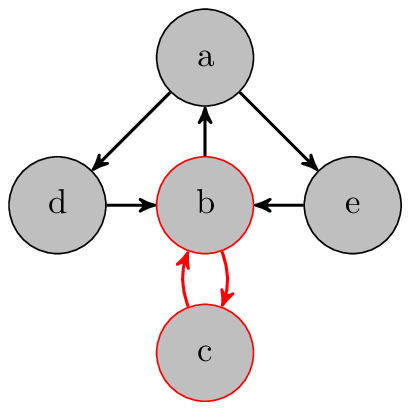}
		\caption{Selection of a Cycle} 
		\label{fig:decomposition:1}
	\end{subfigure}
	\hfill
	\begin{subfigure}{0.32\linewidth}
		\centering
		\includegraphics[width=0.8\linewidth]{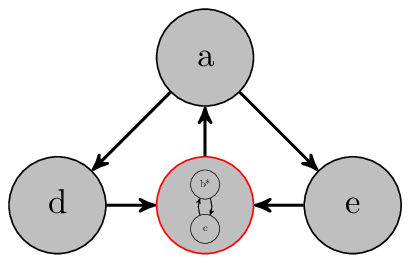}
		\caption{After a Reduction Step}
		\label{fig:decomposition:2}
	\end{subfigure}
	\hfill
	\begin{subfigure}{0.32\linewidth}
		\centering
		\includegraphics[width=0.8\linewidth]{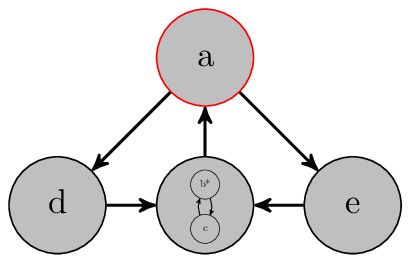}
		\caption{Selection of a Mode to Split}
		\label{fig:decomposition:4}
	\end{subfigure}
	\hfill
	\begin{subfigure}{0.32\linewidth}
		\centering
		\includegraphics[width=0.8\linewidth]{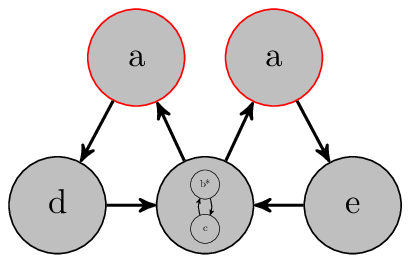}\\\vspace{-1.5cm}
		\caption{After a Mode-splitting step}
		\label{fig:decomposition:5}
	\end{subfigure}
	\hfill
	\begin{subfigure}{0.32\linewidth}
		\centering
		\includegraphics[width=0.8\linewidth]{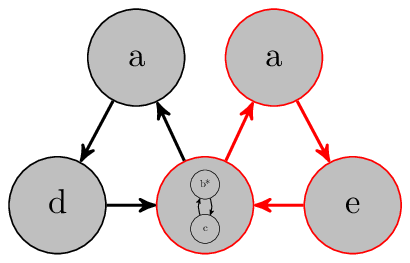}\\\vspace{-1.5cm}
		\caption{Selection of a Cycle}
		\label{fig:decomposition:6}
	\end{subfigure}
	\hfill
	\begin{subfigure}{0.32\linewidth}
		\centering
		\includegraphics[width=0.8\linewidth]{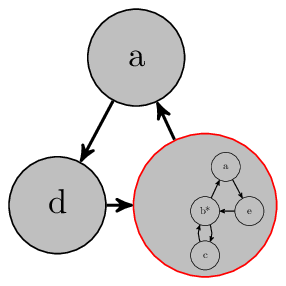}\\\vspace{-1.5cm}
		\caption{After a Reduction Step}
		\label{fig:decomposition:7}
		\label{fig:decomposition:last}
	\end{subfigure}
	\hfill
	\caption{A Sketch of the Decomposition} 
	\label{fig:decomposition}
	\vspace*{-1em}
\end{figure} 
The decomposition technique introduces a so-called
\emph{constraint graph}. In the constraint graph, vertices are labeled with
mode constraints and transition constraints for self-loops, \ie,
\(\Mode_1=\Mode_2\)
while edges are labeled with transition constraints for
non-self-loops, \ie, \(\Mode_1\neq\Mode_2\). Obviously, any solution to the constraint graph is a solution
to \autoref{thm:lyapunov}. The graph structure is exploited in two ways:
\begin{enumerate}[1)]
	\item The constraint graph is partitioned into
		finitely many strongly connected components (SCCs). A trajectory
		entering an SCC of the corresponding hybrid automaton may either
		converge to 0 within the SCC or leave the SCC in finite time. In any
		case, once entered, an SCC might not be entered again. This allows us to
				compute LLFs for each SCC separately.			\item Each SCC is
		further partitioned into (overlapping) cycles. LLFs for modes
		in a cycle can also be computed separately but compatibility 
		-- \wrt constraints on the edges --
		has to be assured somehow. Compatibility can be guaranteed if the
		cycles are examined successively in
		the following way: A cycle is selected and replaced by an
		underapproximation of the feasible set of its constraints, \ie,
		finitely many solutions (candidate LLFs) to the constraints of that
		cycle. Since the constraints describe a convex problem, conical
		combinations of the candidate LLFs satisfy the constraints, too. This
		step is called a \emph{reduction step}. The reduction step collapses all
		vertices that lie only on that cycle and
		replaces references to LLFs in the constraints of adjacent edges by
		conical combinations of the candidate LLFs.
		This allows us to prove stability of each cycle separately while,
		cycle-by-cycle, ensuring compatibility of the feasible sets of
		the (overlapping) cycles.
		\newline
		The reduction step is visualized in \autoref{fig:decomposition:1}
		and \autoref{fig:decomposition:2}: In the former, a cycle is
		selected and in the latter, the cycle is replaced
		by a finite set of solutions of the corresponding optimization problem
		-- visualized by collapsing the cycle into a single vertex.
\end{enumerate}

\begin{figure}[b]
\centering
\hfill
\begin{subfigure}{0.35\linewidth}
	\centering
\includegraphics[height=2.5cm]{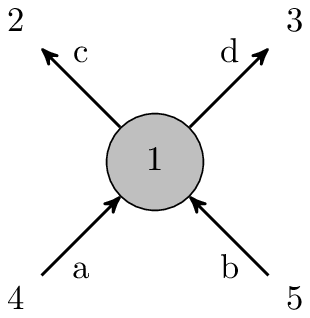}
\caption{Before splitting} 
\label{fig:splitting_before}
\end{subfigure}
\hfill
\begin{subfigure}{0.55\linewidth}
	\centering
\includegraphics[height=2.5cm]{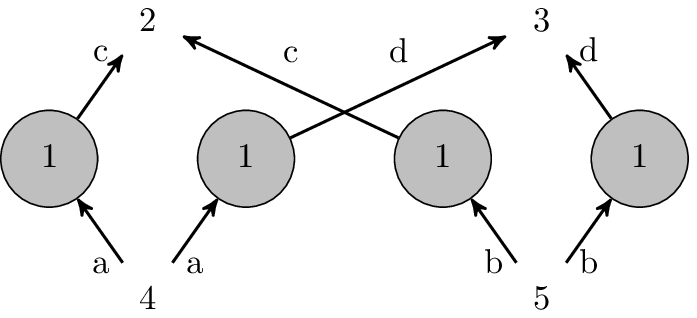}
\caption{After splitting} 
\label{fig:splitting_after}
\end{subfigure}
\hfill
\caption{The Mode-Splitting Step} 
\label{fig:splitting}
\vspace*{-1em}
\end{figure}
The reduction step is more efficient if 
	the cycle is connected to the rest of the graph by at most one vertex.
We call such a cycle an \emph{outer cycle}
and the vertex a \emph{border vertex}.
On one hand, if the graph contains an outer cycle, then 
the cycle can be collapsed into a single vertex which replaces the border
vertex. Thus, the feasible set of the cycle's constraints is replaced by a set of
candidate LLFs. 
On the other hand, if the graph does not contain an outer cycle, then another
step, called a
\emph{mode-splitting step}, is performed. In the mode-splitting step, a
single vertex is
replaced by a copy per pair of incoming and outgoing edges. This is visualized
\autoref{fig:splitting}. In \autoref{fig:splitting_before}, vertex
\texttt{1} is
connected to four other vertices by two incoming and two outgoing edges. In
\autoref{fig:splitting_after}, vertex \texttt{1} is replaced by four copies,
where each one is connected to exactly one incoming and one outgoing edge.
Depending on the order in which vertices are chosen for mode-splitting,
one can make a cycle connected to the rest of the graph by exactly one vertex and then perform a
reduction step. Clearly, the order of mode-splitting and reduction steps does not
only affect the termination of the procedure, but also the size of the
graph and, therefore, the
number of cycles that have to be reduced.
With a good order of reduction and mode-splitting steps, one ends up with a
single cycle for which the following holds: The successful computation of candidate LLFs
implies the existence of a piecewise Lyapunov function for the whole SCC.

Continuing on the example given in \autoref{fig:decomposition}: In
\autoref{fig:decomposition:4}, there are no outer cycles, thus, a mode-splitting
step is performed: the vertex \texttt{a} is selected, copied twice, and
each path is routed through one copy. The result is shown in
\autoref{fig:decomposition:5}. Since the result contains outer cycles, we can
select an outer cycle as in \autoref{fig:decomposition:6} and
perform another reduction step resulting a single cycle being left. \autoref{fig:decomposition:7}
shows the result.

\subsection*{Automatically Computing Lyapunov Functions}

To compute Lyapunov functions needed for decomposition as well as for the
monolithic approaches each Lyapunov function is instantiated by a template
involving free parameters. Using this Lyapunov function templates a constraint
system corresponding to \autoref{thm:lyapunov} is generated. Such a
constraint system is then relaxed by a series of relaxations involving 
\begin{inparaenum}
	\item the so-called \emph{S-Procedure}~\cite{Boyd2004Convex} which
		restricts the constraints to certain regions and
	\item the \emph{sums-of-squares (SOS)}
		decomposition~\cite{Prajna03analysisof}
		which allows us to rewrite the polynomials
		as linear matrix inequalities (LMI).
\end{inparaenum}
These LMIs in turn can be solved by Semidefinite Programming
(SDP)~\cite{Boyd2004Convex}. Instances of solvers are
CSDP~\cite{Borchers99csdp} and
SDPA~\cite{Fujisawa2007sdpa}. These solvers typically use some kind of
interior point methods and numerically approximate a solution.
While this is very fast, such numerical solvers sometimes suffer from
numerical inaccuracies. Therefore, constraints may be \emph{strengthened}
by adding additional \enquote{gaps}. These gaps make the constraints more
robust against numerical issues but sometimes result in the feasible set
becoming empty.

These the gaps further limit the use of the decomposition as each reduction
now \enquote{doubly} shrinks the feasible set: via gaps and
via computing finitely many candidate LLFs.

 \section{Relaxation of the Graph Structure} \label{sec:relaxing}

In this section, we show how the decomposition can be improved by our
graph structure-based relaxation. Consider the \emph{underlying digraph}
\(\Graph=(\SetOfVertices,\SetOfEdges)\) of a hybrid automaton with the set of
vertices \(\SetOfVertices=\SetOfModes\) and the set of edges
\(\SetOfEdges=\set{\tuple{\Mode_1,\Mode_2}}{\exists \tuple{\Mode_1,\GuardSet,\UpdateFunc,\Mode_2} \in \SetOfTrans}\).
Note that the underlying graph has at most a single edge
between any two vertices while the hybrid automaton might have multiple
transitions between two modes.
The \emph{density} of the graph $\Graph$ is the fraction of the number of edges in
the graph and the maximum possible number of edges in a graph of the same
size, \ie,
\(\frac{\card{\SetOfEdges}}{\card{\SetOfVertices}(\card{\SetOfVertices}-1)}\).\footnote{We are referring to the definition of density for directed
graphs.}

The idea is to identify a set of modes of a
hybrid automaton whose graph structure is dense. 
This can, for example, be done
by a clique-finding or dense-subgraph-finding algorithm.
A \emph{clique} is a complete subgraph, \ie, having a density of \(1\).\footnote{Finding the maximum clique is
	\NP-hard. However, a maximum clique is not required, any maximal clique (with
	more than two vertices) is sufficient. Even better as we are interested in dense
	structures only, we can use quasi cliques. A \emph{quasi clique} is a
	subgraph where the 
			density
	is not less than a certain threshold. Thus,
	any greedy algorithm can be used.
}
Our relaxation then rewires the transitions such that the resulting automaton
immediately exhibits a structure well-suited for decomposition. By
\enquote{well-suited,} we mean that the graph structure contains mainly outer
cycles.

The reason, that our relaxation technique plays so well with the decomposition
technique, is as follows: if a hybrid system exhibits a dense graph structure, then the
decomposition results in a huge blow-up. This blow-up is a result of the
splitting step. The splitting step separates vertices shared between cycles, \ie,
if there is more than one vertex shared between two or more cycles, then
multiple copies are created. Thus, the higher the density of the graph
structure is, the higher the blow-up gets. Further, if many cycles share many
vertices -- as in dense graphs -- then whole cycles get copied and each copy
requires solving an optimization problem and underapproximating the
problem's feasible set. In contrast, our relaxation overapproximates the
discrete behavior by putting each vertex in its own cycle and connecting this
vertex by a new \enquote{fake} vertex. This reduces the number of optimization
problems to be solved and the number of feasible sets to be
underapproximated.

In the following, we define the relaxation operator. Then we give an algorithm which
applies the relaxation integrated with decomposition. Finally, we prove
termination and implication of stability of the hybrid automaton which has
been relaxed.

\begin{definition}
\label{def:cs_relax}
The graph structure relaxed hybrid automaton
\(
	\RelaxApp{\HyAt,\SetOfModes_{d}}
	= (\SetOfVars^{\sharp}, \SetOfModes^{\sharp}, \SetOfTrans^{\sharp},
	\MapOfFlow^{\sharp}, \MapOfInvar^{\sharp})
	= \HyAt^{\sharp}
\)
of a hybrid automaton
\(
	\HyAt = (\SetOfVars, \SetOfModes, \SetOfTrans, \MapOfFlow, \MapOfInvar)
\)
\wrt the sub-component \(\SetOfModes_{d} \subseteq \SetOfModes\) is defined as
follows
\begin{align*}
	\SetOfVars^{\sharp}  = \;& \SetOfVars, \\
	\SetOfModes^{\sharp} = \;& \SetOfModes \cupdot \{\Mode_{c}\},\\
\iftrue
	\SetOfTrans^{\sharp} = \;&
		\set{			\Trans{\Mode_{1},\GuardSet,\UpdateFunc,\Mode_{2}} 		}{			\begin{matrix}
			\Trans{\Mode_{1},\GuardSet,\UpdateFunc,\Mode_{2}} \in \SetOfTrans,
			\\
			\{\Mode_{1},\Mode_{2}\} \cap \SetOfModes_{d} = \emptyset
			\end{matrix}
		} \\
		& \bigcup \set{			\begin{matrix}
			\Trans{\Mode_{1},\GuardSet,\IdentFunc,\Mode_{c}}, \\
			\Trans{\Mode_{c},\GuardSet,\UpdateFunc,\Mode_{2}}
			\end{matrix}
		}{			\begin{matrix}
			\Trans{\Mode_{1},\GuardSet,\UpdateFunc,\Mode_{2}} \in \SetOfTrans,
			\\
			\{\Mode_{1},\Mode_{2}\} \cap \SetOfModes_{d} \neq \emptyset
			\end{matrix}
		}, \\
\else
	\SetOfTrans^{\sharp} = \;& (\SetOfTrans \setminus
		\set{			\Trans{\Mode_{1},\GuardSet,\UpdateFunc,\Mode_{2}} 		}{			\{\Mode_{1},\Mode_{2}\} \cap \SetOfModes_{d} \neq \emptyset
		} ) \\
		& \bigcup \set{			\begin{matrix}
			\Trans{\Mode_{1},\GuardSet,\IdentFunc,\Mode_{c}}, \\
			\Trans{\Mode_{c},\GuardSet,\UpdateFunc,\Mode_{2}}
			\end{matrix}
		}{			\begin{matrix}
			\Trans{\Mode_{1},\GuardSet,\UpdateFunc,\Mode_{2}} \in \SetOfTrans,
			\\
			\{\Mode_{1},\Mode_{2}\} \subseteq \SetOfModes_{d}
			\end{matrix}
		} \\
		& \bigcup \set{			\begin{matrix}
			\Trans{\Mode_{1},\GuardSet,\IdentFunc,\Mode_{c}}, \\
			\Trans{\Mode_{c},\GuardSet,\UpdateFunc,\Mode_{2}}
			\end{matrix}
		}{			\begin{matrix}
			\Trans{\Mode_{1},\GuardSet,\UpdateFunc,\Mode_{2}} \in \SetOfTrans,
			\\
			\Mode_{1} \in \SetOfModes_{d},
			\Mode_{2} \not\in \SetOfModes_{d}, \\
			\end{matrix}
		} \\
		& \bigcup \set{			\begin{matrix}
			\Trans{\Mode_{1},\GuardSet,\IdentFunc,\Mode_{c}}, \\
			\Trans{\Mode_{c},\GuardSet,\UpdateFunc,\Mode_{2}}
			\end{matrix}
		}{			\begin{matrix}
			\Trans{\Mode_{1},\GuardSet,\UpdateFunc,\Mode_{2}} \in \SetOfTrans,
			\\
			\Mode_{1} \not\in \SetOfModes_{d},
			\Mode_{2} \in \SetOfModes_{d}, \\
			\end{matrix}
		}, \\
\fi
	\MapOfFlow^{\sharp}\left(\Mode\right) = \;&
		\begin{cases}
			\ZeroFunc &\text{ if } \Mode = \Mode_{c} \\
			\FlowOf{\Mode} &\text{ otherwise, }
		\end{cases} \\
	\MapOfInvar^{\sharp}\left(\Mode\right) = \;&
		\begin{cases}
			\emptyset &\text{ if } \Mode = \Mode_{c} \\
			\InvOf{\Mode} &\text{ otherwise, }
		\end{cases}
\end{align*}
where \(\ZeroFunc : \ContStateSpace \to \PowerSet{\ContStateSpace}\) is a function assigning \(\V{0}\) to each \(\V{x}\in\ContStateSpace\),
\ie, \(\dot{\V{x}}\in\{\V{0}\}\).
\end{definition}

In \(\SetOfTrans^{\sharp}\) in \autoref{def:cs_relax}, we replace each transition
\(\tuple{\Mode_1,\GuardSet,\UpdateFunc,\Mode_2} \in \SetOfTrans\)
connected to at least one mode in \(\SetOfModes_{d}\) with two transitions:
one connecting the old source mode \(\Mode_1\) with the new mode \(\Mode_{c}\)
and the other
connecting \(\Mode_{c}\) with the old target mode \(\Mode_2\).
We call this step a \emph{transition-splitting step} where the result is a pair of
transitions which is called \emph{split transition} and
the set of all split transitions is denoted by \(\Ally\).

Intuitively, the introduced mode \(\Mode_c\) is a dummy mode whose
invariant always evaluates to false and the flow function does not change the
valuations of the continuous variables. Indeed, the mode cannot be entered and
thus, a trajectory taking an ingoing transition must, immediately, take an outgoing
transition. The sole reason to add the mode is changing the
structure of the hybrid system's underlying graph: the new structure
contains mainly cycles that are connected via \(\Mode_c\).

\begin{algorithm}[t]
\SetKwInOut{Input}{input}
\SetKwInOut{Output}{output}
\Input{A hybrid automaton $\HyAt$, a dense sub-component $\SetOfModes_{d}$ of $\HyAt$.}
\Output{The relaxed version of $\HyAt$, a set of split transitions $\Ally$, the central mode $\Mode_{c}$.}
\(\Mode_{c} \leftarrow\) \FuncSty{newMode()}\;
\(\HyAt.\SetOfModes \leftarrow \HyAt.\SetOfModes \cup \{\Mode_{c}\}\)\;
\(\HyAt.\FlowOf{\Mode_{c}} \leftarrow \ZeroFunc\)\;
\(\HyAt.\InvOf{\Mode_{c}} \leftarrow \emptyset\)\;
\(\SetOfTrans \leftarrow \HyAt.\SetOfTrans\)\;
\ForEach{\(\trans = \Trans{\Mode_{1},\GuardSet,\UpdateFunc,\Mode_{2}} \in \SetOfTrans\)}{	\If{\(\{\Mode_{1},\Mode_{2}\} \cap \SetOfModes_{d} \neq \emptyset\)}{		\tcp{split the transitions into two parts}
		\(\trans_1 \leftarrow \Trans{\Mode_{1},\GuardSet,\IdentFunc,\Mode_{c}}\)\;
		\(\trans_2 \leftarrow \Trans{\Mode_{c},\GuardSet,\UpdateFunc,\Mode_{2}}\)\;
		\tcp{replace the transition by the two parts}
		\(\HyAt.\SetOfTrans \leftarrow (\HyAt.\SetOfTrans \setminus \{\trans\} ) \cup
		\{\trans_1,\trans_2\}\)\;
		\tcp{keep account of split transitions}
		\(\Ally \leftarrow \Ally \cup \{(\trans_1,\trans_2)\}\)
	}
}
\caption{The Relaxation Function}
\label{alg:func_relax}
\end{algorithm}

\begin{algorithm}[t]
\SetKwInOut{Input}{input}
\SetKwInOut{Output}{output}
\Input{A relaxed hybrid automaton \(\HyAt\), a set of split transitions \(\Ally\),
	a pair of split transitions \(\tuple{\trans_1,\trans_2}\),
	where \(\trans_1 = \Trans{\Mode_1,\GuardSet,\IdentFunc,\Mode_c}\),
	\(\trans_2 = \Trans{\Mode_c,\GuardSet,\UpdateFunc,\Mode_2}\) and
	\(\Mode_{c}\) is the central mode.
}
\Output{A relaxed hybrid automaton \(\HyAt\) with one split transition being
reconstructed, the set of split transitions \(\Ally\)}
\tcp{reconstruct the original transition}
\(\trans \leftarrow \Trans{\Mode_1,\GuardSet,\UpdateFunc,\Mode_2}\)\;
\tcp{replace the split transition \((\trans_1,\trans_2)\) by \(\trans\)}
\(\HyAt.\SetOfTrans \leftarrow (\HyAt.\SetOfTrans \setminus \{\trans_1,\trans_2\} ) \cup \{\trans\}\)\;
\tcp{update the set of split transitions}
\(\Ally \leftarrow \Ally \setminus \{(\trans_1,\trans_2)\}\)\;
\tcp{remove \(\Mode_{c}\) iff unconnected}
\If{\(\Ally = \emptyset\)}{	\(\HyAt.\SetOfModes \leftarrow \HyAt.\SetOfModes \setminus \{\Mode_{c}\}\)\;
}
\caption{The Reconstruction Function}
\label{alg:func_reconstruct}
\end{algorithm}

\begin{algorithm}[t]
\SetKwInOut{Input}{input}
\SetKwInOut{Output}{output}
\Input{A hybrid automaton \(\HyAt\), a set of modes \(\SetOfModes_{d}\)
	corresponding to a dense subgraph.
}
\Output{\KwSty{stable} if the \(\HyAt\) is stable and \KwSty{failed}
	otherwise.
}
\tcp{relax the graph structure}
\(\HyAt,\Ally,\Mode_{c} \leftarrow\) \FuncSty{relax(\(\HyAt,\SetOfModes_{d}\))}\;
\While{\(\Ally \neq \emptyset\)}{	\tcp{apply decomposition}
	\DataSty{result} $\leftarrow$ \FuncSty{applyDecomposition(\(\HyAt\))}\;
	\If{\DataSty{result} is \KwSty{stable}}{		\Return{\KwSty{stable}}\;	}
	\tcp{apply reconstruction}
			\uIf{\(\exists (\trans_1,\trans_2) \in \Ally: \{\trans_1,\trans_2\} \cap
		\FuncSty{failedSubgraph(\DataSty{result})} \neq \emptyset
	\)}{				\(\HyAt,\Ally \leftarrow\)
		\FuncSty{reconstruct(\(\HyAt,\Ally,\tuple{\trans_1,\trans_2},\Mode_{c}\))}\;
	}
	\Else{		\Return{\DataSty{result}}\;
	}
}
\tcp{apply decomposition on the original automaton}
\DataSty{result} $\leftarrow$ \FuncSty{applyDecomposition(\(\HyAt\))}\;
\Return{\KwSty{\DataSty{result}}}\;
\caption{The Integrated Relaxation and Decomposition Algorithm}
\label{alg:overall}
\end{algorithm}

Next, we show how to integrate decomposition and relaxation.
Pseudo-code of the relaxation function and a reconstruction function -- which
step-by-step reverts the relaxation -- can be found in 
\autoref{alg:func_relax} and \autoref{alg:func_reconstruct}, respectively.
\autoref{alg:overall} gives pseudo-code of the main algorithm.
The main algorithm works as follows:
\begin{inparaenum}[Step 1)]
	\item The function \FuncSty{relax} relaxes the graph structure of the hybrid
		automaton \(\HyAt\) and generates the set of
		split transitions \(\Ally\). 
	\item If the set \(\Ally\) is empty, then call
		\FuncSty{applyDecomposition} with the original automaton and return the result
		-- this function applies the original decompositition technique as described
		in \autoref{sec:decomposition}.
	\item Otherwise, apply
		\FuncSty{applyDecomposition} on the current relaxed form of the
		automaton. If the result is \texttt{stable}, then return the result.
		Otherwise, if
		the original decompositional technique has failed, then it returns a
		\emph{failed subgraph} that is a subgraph
		for which it was unable to find Lyapunov functions.
	\item Choose a split transition
		from the set \(\Ally\) which also belongs to the failed subgraph. It
		is then used to reconstruct a transition from the original hybrid
		automaton.  Then execution is continued with step 2.
	\item If no such split transition exists,
				then the algorithm fails and returns the failed subgraph
		since this failing subgraph will persist in the automaton. 
		Further reverting the relaxation cannot help because no
		split transition is contained in the failed subgraph.
\end{inparaenum}

Next, we prove termination and soundness of the algorithm. Here, soundness 
indicates that a Lyapunov function-based stability certificate
for a relaxed automaton implies stability of the original, unmodified
automaton. In particular, the local Lyapunov functions of the relaxed hybrid
automaton are valid local Lyapunov functions for the original automaton. 

\subsection*{Termination of the Integrated Algorithm}

\begin{theorem}
	The proposed algorithm presented in \autoref{alg:overall} terminates.
\end{theorem}

\begin{proof}
	The function \FuncSty{relax} terminates since the copy of the set of transitions of
	\(\HyAt\) is finite and is not modified in the course of the algorithm.
	The while-loop terminates if either an
	\FuncSty{applyDecomposition} is successful, no pair for
	reconstruction can be identified, or the set \(\Ally\) is empty.
	In the first two cases, the algorithm terminates directly.
	For the last case, we assume that no call to \FuncSty{applyDecomposition} is
	successful and a spilt transition is always found. Then, in each iteration
	of the loop, one edge is removed from \(\Ally\).
	The set \(\Ally\) is
	finite because the relaxation function \FuncSty{relax} splits only
	finitely many edges.
	Thus, the set \(\Ally\) becomes eventually empty.
	Therefore, the loop terminates.
\end{proof}

\subsection*{Preservation of Stability}

\begin{theorem}
\label{thm:preserve_stability}
	For any hybrid automaton \(\HyAt\) and a sub-component
	\(\SetOfModes_{d}\),
	it holds:
	If a family of local Lyapunov functions \((\LF[\Mode])\) 
	proving
	\(\RelaxApp{\HyAt,\SetOfModes_{d}}\) to be GAS exists, then
	there exists a family of local Lyapunov functions for
	\(\HyAt\) proving \(\HyAt\) to be GAS.
\end{theorem}

\begin{proof}
	Given a hybrid automaton \(\HyAt = (\SetOfVars, \SetOfModes, \SetOfTrans,
	\MapOfFlow, \MapOfInvar)\). 
	Let
	\(
		\RelaxApp{\HyAt,\SetOfModes_{d}}
		= (\SetOfVars^{\sharp}, \SetOfModes^{\sharp}, \SetOfTrans^{\sharp},
		\MapOfFlow^{\sharp}, \MapOfInvar^{\sharp})
		\allowbreak
		= \HyAt^{\sharp} 
	\), 
	be a graph structure-relaxed version of \(\HyAt\) where 
	\(\SetOfModes_{d}\subseteq\SetOfModes\) is the sub-component of \(\HyAt\) that has been
	relaxed. Further, let \((\LF[\Mode])\) be the family of local Lyapunov
	functions that prove stability of \(\HyAt^{\sharp}\) and let \(\Ally\) be
	the set of split transitions -- some transition may have been
	reconstructed.
	Now, it must be shown that \(\LF[\Mode]\) are
	valid Lyapunov functions for \(\HyAt\).\newline
	The mode constraints of \autoref{thm:lyapunov} trivially hold, since
	\(\Relax\) alters neither the flow functions nor the invariants, \ie,
	\(
		\forall \Mode\in\SetOfModes :
		\MapOfFlow^{\sharp}\left(\Mode\right) = \FlowOf{\Mode} \land
		\MapOfInvar^{\sharp}\left(\Mode\right) = \InvOf{\Mode}
	\).
	The transition constraint also holds for all transitions that are not
	altered by \(\Relax\) or have been reconstructed, \ie,
	\(\SetOfTrans\cap\SetOfTrans^{\sharp}\).
	Now assume that
	\(\trans\in\SetOfTrans\setminus\SetOfTrans^{\sharp}\) is an arbitrary
	transition for which the transition constraint does not hold. We show that this leads to a
	contradiction.
	Due to the definition of \(\Relax\)
	all transition in \(\SetOfTrans\setminus\SetOfTrans^{\sharp}\)
	are split transitions and there is a corresponding pair in \(\Ally\).
	Let \(\tuple{\trans_1,\trans_2}\in\Ally\) be the pair corresponding to
	\(\trans=(\Mode_1,\GuardSet,\UpdateFunc,\Mode_2)\).
	Since \((\LF[\Mode])\) is a valid family of
	local Lyapunov function for \(\HyAt^{\sharp}\), the transition constraint holds for all
	transitions in \(\SetOfTrans^{\sharp}\). In particular, the transition constraint holds
	for
		\(\trans_1=(\Mode_1,\GuardSet,\IdentFunc,\Mode_c)\)
	and
		\(\trans_2=(\Mode_c,\GuardSet,\UpdateFunc,\Mode_2)\). Thus,
\iftwocolumn
	\begin{align*}
			& \holds[\forall x \in \GuardSet]{			\LFof[\Mode_{c}]{\IdentOf{x}} \leq \LFof[\Mode_{1}]{x}} \\
		\land & \holds[\forall x \in \GuardSet]{			\LFof[\Mode_{2}]{\UpdateOf{x}} \leq \LFof[\Mode_{c}]{x}}\text{.}
	\end{align*}
\else
	\begin{align*}
		\holds[\forall x \in \GuardSet]{			\LFof[\Mode_{c}]{\IdentOf{x}} \leq \LFof[\Mode_{1}]{x}}
		\land \holds[\forall x \in \GuardSet]{			\LFof[\Mode_{2}]{\UpdateOf{x}} \leq \LFof[\Mode_{c}]{x}}\text{.}
	\end{align*}
\fi
	It follows, that
	\[
		\holds[\forall x \in \GuardSet]{
			\LFof[\Mode_{2}]{\UpdateOf{x}}
				\leq \LFof[\Mode_{c}]{x}
				\leq \LFof[\Mode_{1}]{x}\text{.}
		}
	\]
	Therefore, the transition constraint holds for \(\trans\). But this contradicts the assumption.
	\(\lightning\)
\end{proof}

While \autoref{thm:preserve_stability} shows that stability of the relaxed
automaton yields stability of the original automaton, the contrary is not
true. \autoref{fig:example_unstable} shows a hybrid system where the
relaxation renders the system unstable. This example exploits that
the relaxation may introduce spurious trajectories. This happens if there are
transitions with overlapping guard sets connected to the central mode
\(\Mode_c\). A trajectory of the relaxed
automaton might then take the first part of a split transition to the
central mode \(\Mode_c\) and continues with the second part of a
different split transition. A transition corresponding to this behavior
might not exist in the unmodified hybrid automaton. While this does not
render our approach being incorrect, it may lead to difficulties since these
extra trajectories have to be GAS, too. In case of the system in
\autoref{fig:example_unstable}, new trajectories are introduced which allow
a trajectory to jump back from the mode \texttt{L} to \texttt{H} by taking the transitions
\(\trans_1,\trans_2\).
This behavior corresponds to leaving
\texttt{L} by the right self-loop and entering \texttt{H} by the left
self-loop, which is obviously impossible. However, due to the update the
value of \(x\) might increase as \(1 + 0.01(x-1)(x-10)>1\) for \(x<1\).

In general our relaxation introduces conservatism which is again reduced
step-by-step by the reconstruction. The degree of conservatism highly depends
on the guards of the transitions since the central mode relates all LLFs
of modes in \(\SetOfModes_d\). Therefore, if more guards are overlapping,
more LLFs have to be compatible even if not needed in the original automaton.

One possibility to counter-act this issue is to introduce a
new continuous variable in the relaxed automaton which is set to a unique
value per split transition: the update function of the first part of a
split transition sets the value used to guard the second part of the
transition.
Indeed, this trick discards any spurious trajectories
for the price of an additional continuous variable. However, since the
values of
that variable are somewhat artificial, a Lyapunov function may not
make use of that variable. Thus, this trick will not ease satisfying the
conditions of the Lyapunov theorem in general.

\begin{figure}[h]
\centering
\begin{subfigure}{0.45\linewidth}
	\includegraphics[width=\linewidth]{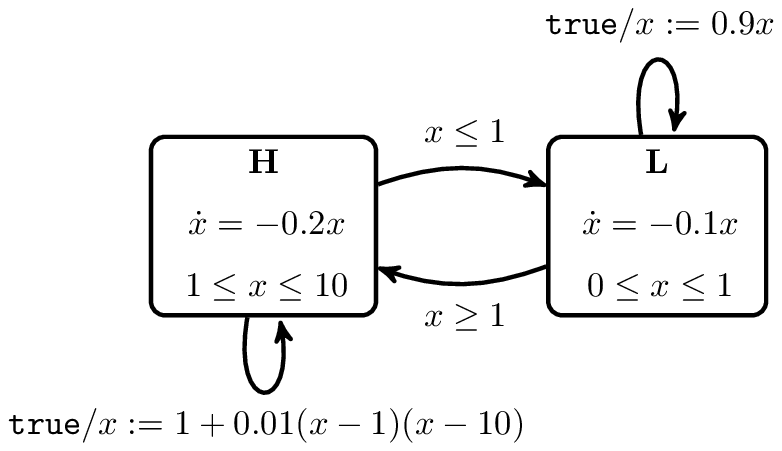}
\caption{The unmodified (stable) version.} 
\label{fig:example_unstable_orig}
\end{subfigure}
\hfill
\begin{subfigure}{0.45\linewidth}
	\includegraphics[width=\linewidth]{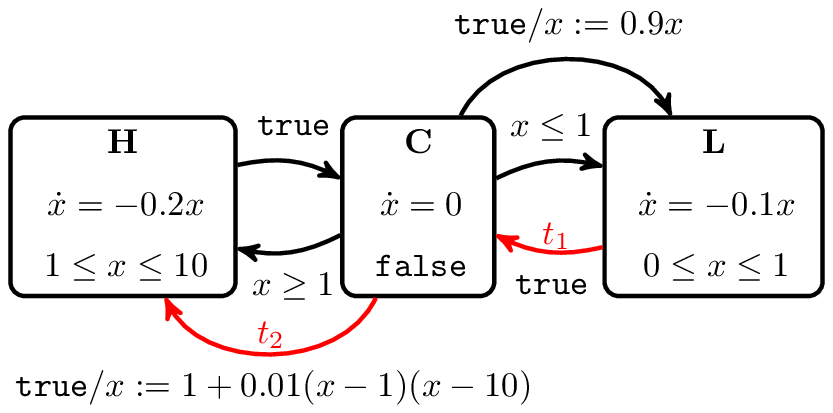}
\caption{The relaxed (unstable) version.} 
\label{fig:example_unstable_relaxed}
\end{subfigure}
\caption{A Hybrid System; Unstable after Relaxation}
\label{fig:example_unstable}
\vspace*{-1em}
\end{figure}
 \section{Application of the Relaxation} \label{sec:experiments}

In this section we present three examples where the graph structure-based
relaxation suggested in this paper improves the application of the
decomposition technique. 
The first example deals with the automatic cruise controller (ACC) of
\cite{Oehlerking2011thesis}.
The second example is the fully connected digraph \(K_3\). The \(K_3\) does
not represent
a concrete hybrid automaton but a potential graph structure of a hybrid
automaton. 
The last example is a spidercam. Here, the graph is not as fully connected as
the \(K_3\) example, but its density is already too high to apply
decomposition directly. 

We have implemented the decomposition and relaxation in python.
\autoref{tab:cycle_data} gives the graph properties and a comparison of the
number of reduction steps required by the decomposition with and without
relaxation (in the best case). 
The given data was obtained without
actually computing Lyapunov functions focusing on the graph related part of
the decomposition. In fact, computing Lyapunov functions for the spidercam via
decomposition without our relaxation fails after 18 steps.

\begin{table*}
	\scriptsize
\begin{tabularx}{\linewidth}{cXX|XXc|Xc}
	\hline
	&&&\multicolumn{3}{c|}{Decomposition}& \multicolumn{2}{c}{With Relaxation}\\
	Graph Structure & Nodes (\(n\)) & Edges & Reductions & Mode-Splittings &
	Time & Reductions & Time\\
	\hline
	directed \(K_1\) & 1 & 0 & 0 & 0 & 0.04s & 0 & 0.04s \\
	directed \(K_2\) & 2 & 2 & 1 & 0 & 0.04s & 2 & 0.04s \\
	directed \(K_3\) & 3 & 6 & 6 & 4 & 0.21s & 3 & 0.05s \\
	directed \(K_4\) & 4 & 12 & 47 & 25 & 1.15s & 4 & 0.05s \\
	directed \(K_5\) & 5 & 20 & 1852 & 352 & 13h22m & 5 & 0.05s \\ 		Spidercam & 9 & 32 & 753 & 287 & 1h46m & 9 & 0.06s\\
	Cruise Controller & 6 & 11 & 7 & 6 & 0.060s & 6 & 0.06s\\
	\hline
\end{tabularx}
\caption{Comparison of the Decomposition with and without Relaxation
			}
\label{tab:cycle_data}
\vspace*{-1em}
\end{table*}

\subsection*{Example 1: The Automatic Cruise Controller (ACC)}

\iftwocolumn
\begin{figure}[!h]
	\centering
	\includegraphics[width=\linewidth]{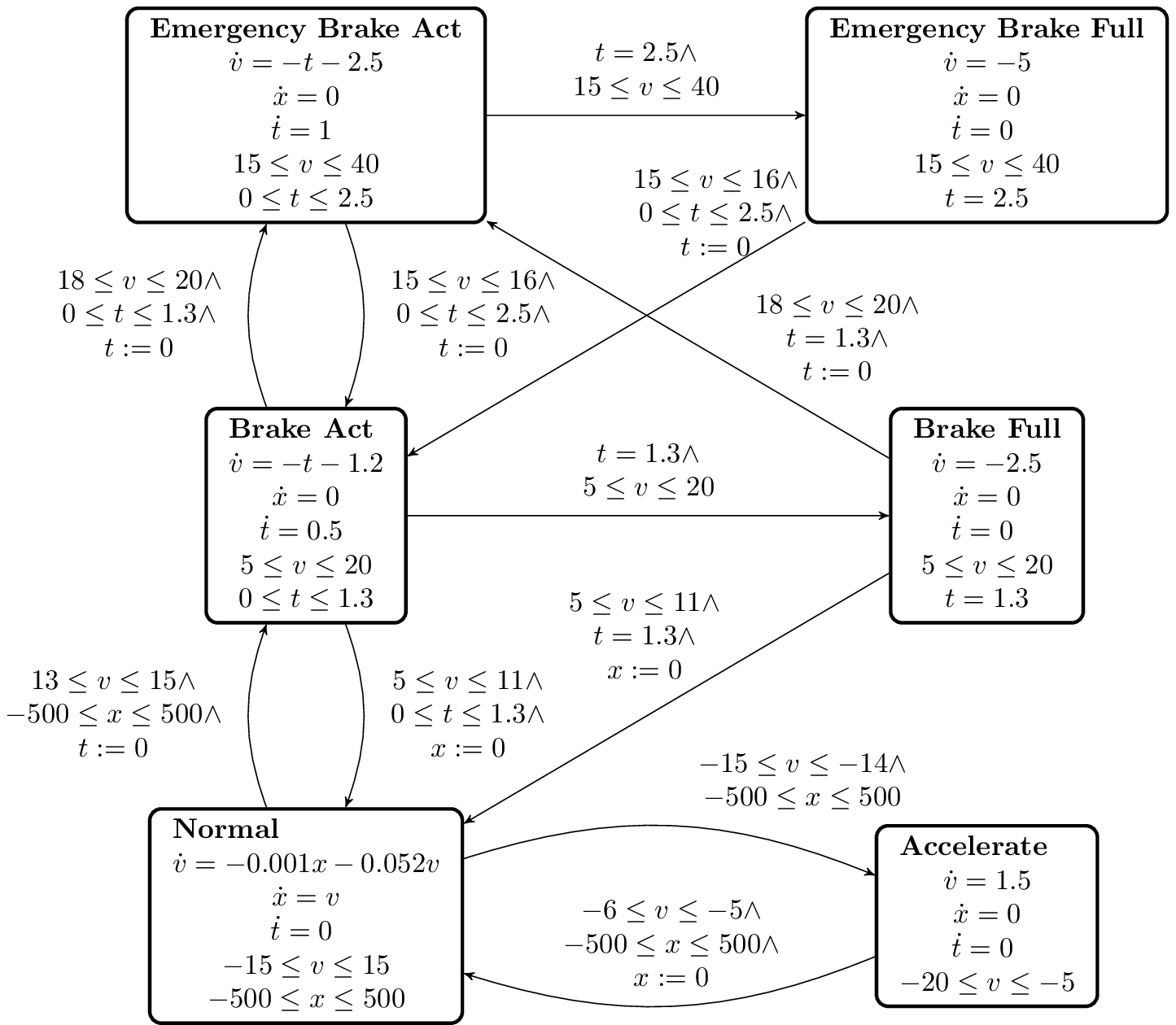}
	\caption{The Automatic Cruise Controller\cite{Oehlerking2011thesis}} 
	\label{fig:example_cruise_controller}
	\vspace*{-1em}
\end{figure}\else
\begin{figure}
	\begin{minipage}{0.39\linewidth}
	\begin{center}
		\begin{subfigure}{\linewidth}
		\centering
		\includegraphics[width=0.5\linewidth]{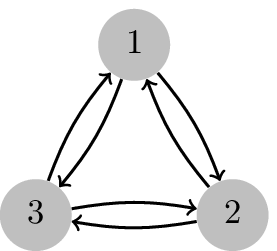}
		\caption{original} 
		\label{fig:example_k3-normal}
		\end{subfigure}
		\newline
		\begin{subfigure}{\linewidth}
		\centering
		\includegraphics[width=0.4\linewidth]{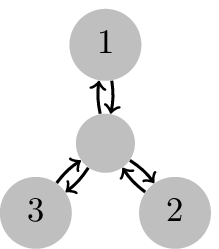}
		\caption{relaxed} 
		\label{fig:example_k3-relaxed}
		\end{subfigure}
		\caption{The \(K_3\).} 
		\label{fig:example_k3}
	\end{center}
	\end{minipage}
	\hfill
	\begin{minipage}{0.59\linewidth}
		\centering
		\includegraphics[width=\linewidth]{figures/cruise_controller}
		\caption{The Automatic Cruise Controller \cite{Oehlerking2011thesis}} 
		\label{fig:example_cruise_controller}
	\end{minipage}
	\vspace*{-1em}
\end{figure}\fi

The automatic cruise controller (ACC) regulates the velocity of a vehicle. 
\autoref{fig:example_cruise_controller} shows the controller as an automaton.
The task of the controller is to approach a user-chosen velocity -- indeed
the variable \(v\) represents the velocity relative to the desired velocity.

The ACC is globally asymptotically stable. It
can be proven stable using the original decomposition technique (cf.
\cite{Oehlerking2011thesis,MohlmannT13stabhyli}). Indeed,
the graph structure is sparse and thus, already well-suited for applying the
decomposition technique directly. In fact, only one more cycle needs to be
reduced compared to decomposition after relaxation
(cf. \autoref{tab:cycle_data}). Even though the relaxation is not needed
here, it also does not harm, though, it may be used for sparse graphs structures,
too.

\subsection*{Example 2: The directed \(K_3\)}

\iftwocolumn
\begin{figure}[h]
	\centering
	\begin{subfigure}{0.40\linewidth}
	\includegraphics[width=\linewidth]{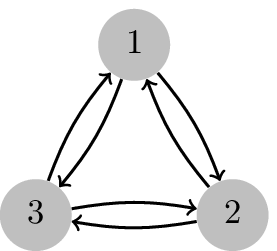}
	\caption{The unmodified \(K_3\)} 
	\label{fig:example_k3-normal}
	\end{subfigure}
	\hfill
	\begin{subfigure}{0.40\linewidth}
	\includegraphics[width=0.8\linewidth]{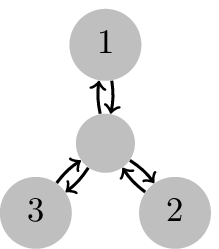}
	\caption{The relaxed \(K_3\)} 
	\label{fig:example_k3-relaxed}
	\end{subfigure}
	\caption{The \(K_3\)} 
	\label{fig:example_k3}
	\vspace*{-1em}
\end{figure}\fi

The directed \(K_n\) is a fully connected digraph with \(n\) nodes. The
\(K_3\) as well as a relaxed version of it is shown in
\autoref{fig:example_k3}.
In a fully connected
digraph, there is a single edge from each node to each other node, resulting in
a total number of $n(n-1)$ edges. The number of cycles, the
decomposition technique has to reduce, grows very fast with \(n\) which can
be seen in \autoref{tab:cycle_data}. In comparison, the number of cycles in
the relaxed version of the graph grows linearly with \(n\), assuming that the
edges can be concentrated\footnote{With \enquote{concentrating edges,} we mean that edges with the same
	source and target node are handled as a single edge for the cycle finding
	algorithm.
}.
Otherwise, after the relaxation, each original node has
\(n-1\) incoming and \(n-1\) outgoing edges where each edge connects the node
with the central node \(\Mode_c\). Each such combination forms a cycle between
\(\Mode_c\) and an original mode, giving a total of \(n(n-1)(n-1)\) cycles in
the worst case. This cubic growth is still much less than the number of
reductions without relaxation.

Such a graph might not be the result of a by-hand designed system but might be
the outcome of a synthesis or an automatic translation. However, the fast
growth of the cycles also indicates the high number of reduction and
therefore underapproximations.

\subsection*{Example 3: The Spidercam}

A spidercam is a movable robot equipped with a camera. It is used at sport
events such as a football matches.
The robot is connected to four cables. Each cable is attached to a motor
that is placed high above the playing field in a corner of a stadium. By winding
and unwinding the cables -- and thereby controlling
the length of the cables, -- the spidercam is able to reach nearly any position
in the three-dimensional space above the playing field.
\autoref{fig:example_spidercam} shows a very simple model of such a
spidercam in the plane. The target is to stabilize the camera at a certain
position. The continuous variables \(x\) and \(y\) denote the distance
relative to the desired position on the axis induced by the cables.

\iftwocolumn
\begin{figure}[h]
\centering
\includegraphics[width=\linewidth]{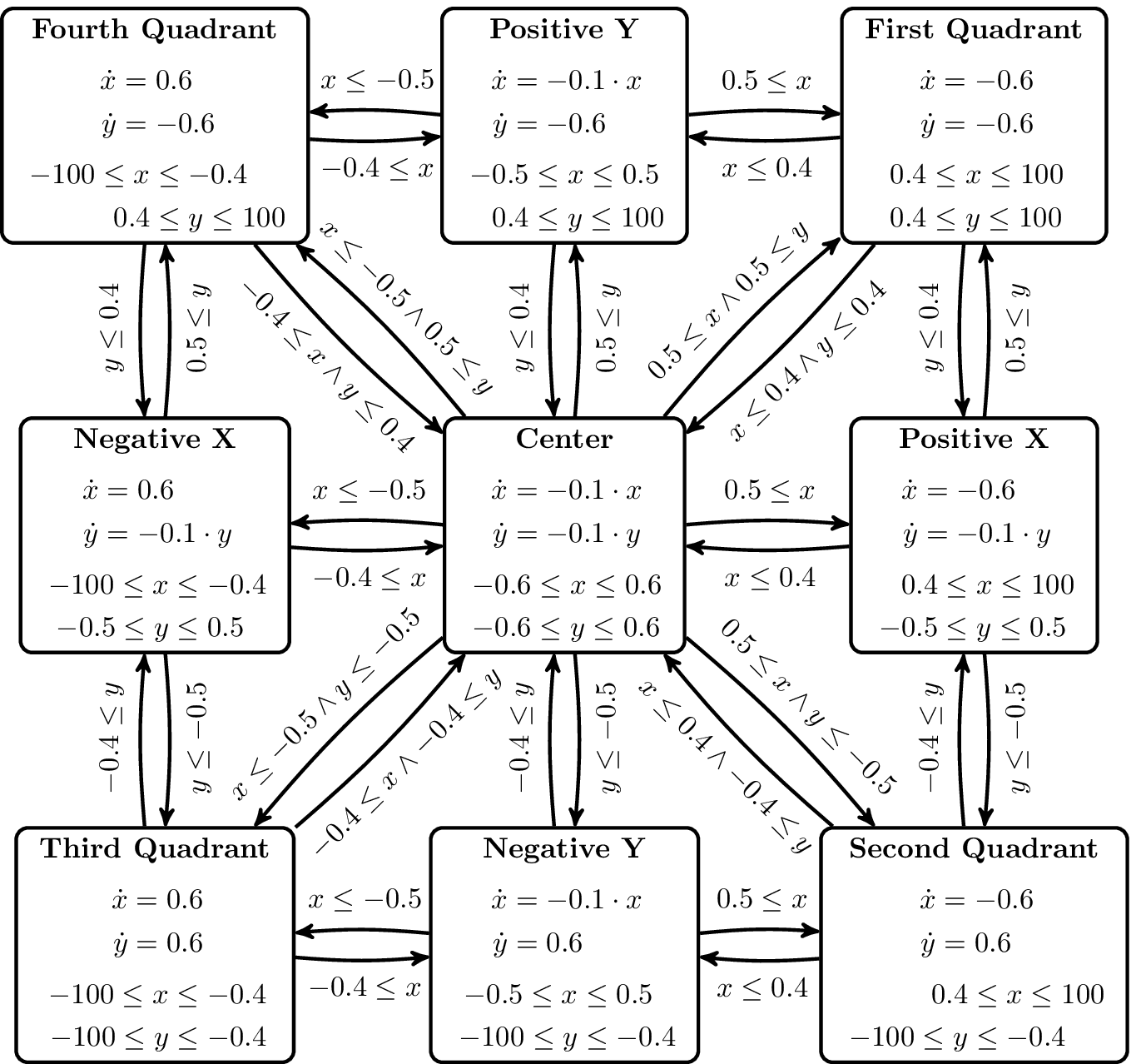}
\caption{The Simple Planar Spidercam} 
\label{fig:example_spidercam}
\vspace*{-1em}
\end{figure}\else
\begin{figure}[h]
	\centering
	\begin{subfigure}{0.49\linewidth}
	\includegraphics[width=\linewidth]{figures/example_spidercam}
	\caption{unmodified} 
	\label{fig:example_spidercam-normal}
	\end{subfigure}
	\hfill
	\begin{subfigure}{0.499\linewidth}
	\includegraphics[width=1.02\linewidth]{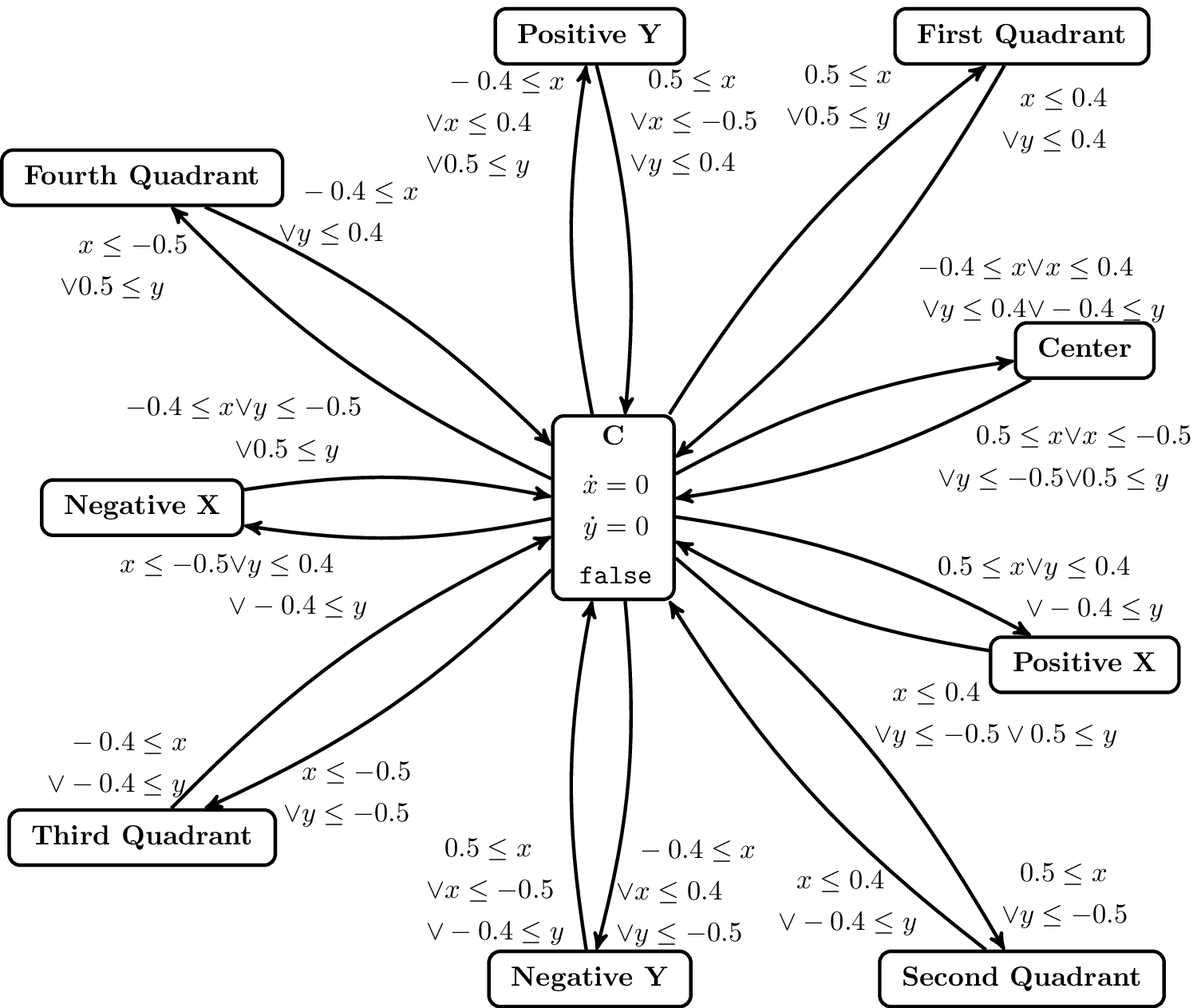}
	\caption{relaxed} 
	\label{fig:example_spidercam-relaxed}
	\end{subfigure}
	\caption{The Simple Planar Spidercam} 
	\label{fig:example_spidercam}
	\vspace*{-1em}
\end{figure}\fi

In the model, we assume a high-level control of the motor engines, \ie, the
movement is on axis \(\dot{x}\) and \(\dot{y}\) instead of a low-level
control of each individual motor. The model has nine modes: one mode that controls
the behavior while being close to the
desired position, four modes corresponding to nearly straight movements along one of
the axes and four modes cover the quadrants between the axes. The maximal
velocity in the direction of each axis is limited from above by
\(0.6\frac{m}{s}\).
Thus, in the four modes corresponding to the quadrants, the movement in each
direction is at full speed. In the four modes corresponding to the axes, the
movement on the particular axis is at full speed while the movement orthogonal to the
axis is proportional to the distance. In the last mode, the speed in both
directions is proportional to the distance.

The spidercam is globally asymptotically stable which can be proven
fully automatically.
However, it is not possible to 
obtain a piecewise Lyapunov function via decomposition
without relaxation due to accumulating underapproximations of the partial
solutions and the high number of cycles that have to be reduced.\footnote{We used the implementation in \Stabhyli~\cite{MohlmannT13stabhyli}
	which currently does not contain strategies to handle the situation where
	no reduction is possible. The current implementation would then simply
	fail.
	Even though, it is theoretically possible to
	perform some form of backtracking, it is hard to decide which
	underapproximation must be refined.
} 
The reason is that each time a cycle is reduced, the feasible set of a
subproblem is underapproximated by a finite set of solutions which finally
results in a feasible set becoming empty and no LLFs can be found.

In contrast, relaxing the graph structure followed by applying the
decomposition is successful immediately. In particular, no reconstruction step
is required.

 \section{Summary} \label{sec:summary}

We have presented a relaxation technique based on the graph structure of a
hybrid automaton. The relaxation exploits super-dense switching or cascaded
transitions to modify the transitions of the hybrid automaton in a way that
improves the decompositional proof technique of \cite{OehlerkingT09}. The idea
is to re-route every transition through a new \enquote{fake} node. Thus, if
in the original automaton a single transition is taken, then the relaxed
automaton has to take the cascade of two transitions to achieve the same
result. However, the relaxed automaton's graph structure is better suited
towards decomposition. Furthermore, the procedure can be automated which is very much desired
as our focus is the automation of
Lyapunov function-based stability proofs.
Furthermore, in \autoref{sec:experiments}, we
successfully employed the proposed technique in some examples.

The decompositional proof
technique is particularly well-suited to prove stability of large-scale hybrid
systems because it allows:
\begin{inparaenum}[1.]
	\item to decompose a monolithic proof into several smaller subproofs,
	\item to reuse subproofs after modifying the hybrid system, and
	\item to identify critical parts of the hybrid automaton.
\end{inparaenum}
All these benefits are not available when 
the hybrid system exhibits a very dense graph structure of the automaton because that would
lead to an enormous number of computational steps required in the
decomposition. The proposed relaxation overcomes these matters in the best
case. If the relaxation is too loose, then our technique
falls back to step-by-step reconstruct the original automaton. Each step
increases the effort needed for the
decomposition until a proof succeeds or ultimately -- in the worst case --
the original automaton gets decomposed.
Future research will include a tighter coupling of the decomposition
and our relaxation approach. A first step will be to not discard the progress made by
the decomposition but reuse the \enquote{gained knowledge}. Doing so will
greatly reduce the
computational effort.

\bibliographystyle{bst/eptcs.bst}
\end{document}